\documentclass[conference]{IEEEtran}
\IEEEoverridecommandlockouts
\usepackage[noend]{algpseudocode}
\usepackage{algorithm}

\usepackage{graphicx}
\graphicspath{{./figs/}}

\usepackage{amsmath}
\usepackage{amssymb}
\usepackage{amsfonts}
\usepackage{amsthm}
\usepackage[mathcal]{eucal}
\usepackage{mathrsfs}
\usepackage{booktabs}
\usepackage{enumitem}
\usepackage{multirow}
\usepackage{bm}
\usepackage{threeparttable}
\usepackage{makecell}
\usepackage{placeins}
\usepackage{pifont}
\usepackage{circledsteps}
\usepackage{xspace}
\usepackage{xcolor,colortbl}
\usepackage{tikz}
\usepackage[normalem]{ulem}
\usepackage{inconsolata}
\usepackage[frozencache,cachedir=.]{minted}
\usepackage{subcaption}
\newtheorem{theorem}{Theorem}

\makeatletter
\renewcommand{\footnoterule}{%
    \kern -3pt
    \hrule \@width 0.28\columnwidth \@height 0.4pt
    \kern 2pt
}
\makeatother

\usepackage{cite}
\usepackage{url}
\usepackage[bookmarks=false]{hyperref}

\definecolor{CUpurple}{RGB}{136,43,142}

\hypersetup{
    colorlinks = true,
    citecolor  = CUpurple,
    linkcolor  = CUpurple,
    urlcolor   = CUpurple
}

\usepackage{cleveref}
\Crefformat{figure}{Fig.~#2#1#3}
\Crefname{subfigure}{Fig.}{Figs.}
\Crefname{figure}{Fig.}{Figs.}
\Crefformat{table}{TABLE~#2#1#3}

\definecolor{CUHKorange}{RGB}{244,106,18}
\definecolor{CUHKblue}{RGB}{0,111,190}
\definecolor{CUHKgreen}{RGB}{0,127,128}
\definecolor{CUHKred}{RGB}{228,46,36}
\definecolor{CUHKyellow}{RGB}{198,148,34}
\definecolor{CUHKdark}{RGB}{114,44,114}
\definecolor{CUHKmiddle}{RGB}{144,44,144}
\definecolor{CUHKlight}{RGB}{167,44,167}
\definecolor{CUHKpurple}{RGB}{117,15,109}
\definecolor{CUHKgold}{RGB}{221,163,0}
\definecolor{CUHKribbon}{RGB}{244,223,176}
\definecolor{CUHKblack}{RGB}{34,24,21}

\newcommand{\mathbbm}[1]{\text{\usefont{U}{bbm}{m}{n}#1}}
\newcommand{\minisection}[1]{\vspace{.02in}\noindent{\textbf{#1}}.}

\newcommand*\circled[1]{%
    \tikz[baseline=(char.base)]{%
        \node[shape=circle,color=blue,fill=blue!20,draw,inner sep=0.6pt] (char) {#1};%
    }%
}
\DeclareMathAlphabet\mathbfcal{OMS}{cmsy}{b}{n}

\newcommand{\CPPL}{\textit{CPPL}\xspace}
\newcommand{\CPPLIR}{\text{CPPL IR}\xspace}

\begin{document}

\title{
    \textit{CPPL}: A \underline{C}ircuit \underline{P}rompt \underline{P}rogramming \underline{L}anguage 
}

\iftrue
\author{
    Shuo Yin$^1$,
    Yihe Wang$^2$,
    Lancheng Zou$^1$,
    Xufeng Yao$^1$,
    Tinghuan Chen$^2$,
    Chen Bai$^{3,\dagger}$, \\
    Zhengrong Wang$^1$,
    Tsung-Yi Ho$^1$,
    Bei Yu$^{1,\dagger}$\thanks{$^\dagger$Corresponding authors.},\\
    $^1$The Chinese University of Hong Kong \\
    $^2$The Chinese University of Hong Kong (Shenzhen) \\
    $^3$Fudan University
}
\fi

\maketitle
\pagestyle{plain}

\begin{abstract}
Large language models (LLMs) have shown promise in register-transfer level (RTL) design automation, but direct RTL generation remains difficult to validate, optimize, and integrate with compiler-based hardware design flows.
Hardware compiler infrastructures such as CIRCT provide typed intermediate representations, legality checks, and optimization passes, yet current LLMs struggle to emit raw compiler IR because of MLIR syntax, SSA discipline, dialect-specific operations, and strict width constraints.
This paper presents \CPPL, a compiler-mediated design framework that turns LLM-assisted hardware generation into a statically checkable frontend problem rather than an unconstrained RTL text-generation task.
\CPPL combines a Python frontend DSL for declaring module interfaces and hierarchy with \CPPLIR, a JSON-based circuit IR designed to expose compiler-visible structure while remaining accessible to LLMs.
The compiler infers operation widths from declared module ports, validates generated IR, checks hierarchy and port bindings, and deterministically lowers the result to CIRCT for synthesizable Verilog generation.
On the RTLLM benchmark, \CPPL improves functional correctness over direct Verilog and direct CIRCT IR generation, while CIRCT optimization reduces post-synthesis AIG node counts.
These results show that a compiler-mediated interface can make LLM-assisted hardware design more reliable, analyzable, and amenable to backend optimization.
\CPPL is available at \url{https://github.com/SawyDust1228/CPPL}.

\end{abstract}

\section{Introduction}
\label{sec:introduction}

Large language models (LLMs)~\cite{qwen3,deepseek,gpt} are increasingly used for register-transfer level (RTL) design automation, including Verilog generation, repair, and verification-oriented coding~\cite{verilogeval,rtllm,origen,betterv,rtlcoder}.
These systems lower the barrier for hardware design by translating natural-language specifications into executable hardware descriptions.
However, most existing LLM4RTL flows follow an end-to-end generation paradigm: the model emits final RTL text, which is then checked by simulators or synthesis tools.
This paradigm is convenient, but from a design automation perspective it leaves three recurring problems unresolved.
First, generated RTL can be syntactically invalid or incompatible with downstream tools.
Second, syntactically valid RTL often fails functional tests because the model must simultaneously reason about behavior, structure, widths, and corner cases.
Third, end-to-end RTL generation exposes little intermediate structure to hardware compiler infrastructures that already provide typed representations, legality checks, canonicalization, and optimization.

Hardware compiler infrastructures such as CIRCT~\cite{circt} provide a natural way to address these issues.
CIRCT represents circuits in MLIR-based dialects and can lower well-formed intermediate representations into synthesizable Verilog while applying standard compiler transformations.
In principle, an LLM could generate CIRCT IR directly and thereby benefit from compiler verification and optimization.
In practice, this is difficult.
CIRCT IR exposes MLIR syntax, SSA naming, dialect-specific operation constraints, and strict type requirements.
Our profiling results in \Cref{sec:background-and-motivation} show that even strong commercial models generate CIRCT IR with much lower correctness than Verilog, despite receiving format-specific prompts.
This gap suggests that compiler-backed LLM hardware generation needs a frontend interface that is more structured than natural language or raw RTL, yet easier for LLMs to produce than low-level compiler IR.

\begin{figure}[!tb]
     \centering
    \includegraphics[width=\linewidth]{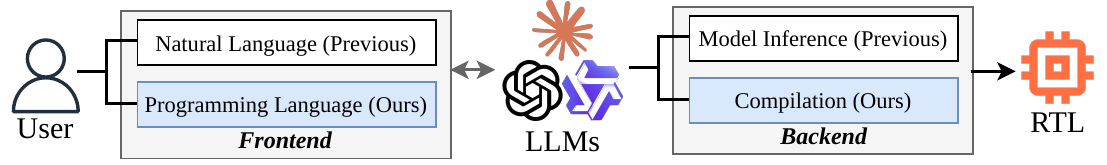}
    \caption{\textit{CPPL} combines LLM-based generation with compiler-mediated circuit construction and optimization.}
    \label{fig:intro}
\end{figure}

This paper presents \CPPL, a compiler-mediated framework for LLM-assisted hardware generation.
\CPPL introduces a Python-based frontend DSL that captures module interfaces and structural hierarchy explicitly, while leaving implementation intent in an LLM-friendly form.
The frontend elaborates fixed ports, module instances, and connection structure before LLM generation, preventing the model from freely inventing incompatible interfaces or hierarchy.
For the model-generated part, \CPPL uses \CPPLIR, a JSON-based intermediate representation that encodes circuit operations in a regular schema.
\CPPLIR is designed to retain compiler-visible circuit structure while avoiding raw MLIR syntax; the compiler performs syntax validation, width inference, structural checks, and deterministic lowering to CIRCT IR.
The resulting CIRCT program is then compiled and optimized to Verilog.
In this way, \CPPL shifts LLM generation from unstructured RTL or raw CIRCT IR toward a typed, structurally constrained frontend that can be checked before backend code generation.
\Cref{fig:intro} summarizes this shift.
Unlike prior LLM4RTL flows that rely on natural-language prompts and direct model inference to produce RTL, \CPPL asks users to describe hardware through a programming-language frontend and delegates RTL construction to a compiler backend.
This separation keeps interface and hierarchy information explicit before LLM generation, and moves legality checking, refinement, and optimization into the compilation path.

Our evaluation on the RTLLM benchmark~\cite{rtllm} shows that direct Verilog generation achieves high syntax correctness but still suffers from a substantial functionality gap, while direct CIRCT IR generation is significantly less reliable.
\CPPL closes this gap by assigning interface construction, hierarchy elaboration, legality checking, and type recovery to the compiler while leaving the model to generate a constrained circuit representation.
Across the evaluated models, \CPPL improves functional correctness over both direct Verilog generation and direct CIRCT IR generation.
We also evaluate synthesis quality using post-\texttt{aigmap} node counts and show that CIRCT optimization passes can reduce synthesized circuit size.

This paper makes the following contributions:
\begin{itemize}[leftmargin=*, itemindent=0pt, listparindent=0pt]
    \item We identify the mismatch between LLM generation capabilities and raw CIRCT IR generation, showing that compiler IR syntax and semantics remain difficult for current LLMs.
    \item We propose \CPPL, an open-source compiler-mediated hardware generation framework that combines a Python frontend DSL, a JSON-based circuit IR, static checking, and CIRCT-based lowering.
    \item We formalize the structural preservation and width inference principles used by \CPPL to keep module hierarchy, port bindings, and operation types consistent during generation.
    \item We evaluate \CPPL on RTLLM and demonstrate improved functional correctness and synthesis-level compactness compared with direct generation baselines.
\end{itemize}

\section{Background \& Motivation}
\label{sec:background-and-motivation}

\subsection{Related Works on LLM4RTL}
\label{sec:related-works}

Recent LLMs have shown promising capabilities in generating RTL code from high-level descriptions, attracting significant attention in the research community.
Several approaches have been proposed to improve LLM-based RTL generation:~\cite{rtlcoder,rtlpp} explore LLM fine-tuning and graph embedding techniques to enhance generation performance, while~\cite{origen,betterv} introduce code-to-code alignment methods to improve the quality of generated RTL code.
\cite{mage,spec2rtl} leverage agent systems to decompose complex hardware generation tasks into manageable subtasks, improving RTL generation accuracy.
Beyond code generation, LLMs also play crucial roles in other RTL-related tasks.
\cite{symrtlo, rtlrewriter, aspen} optimize generated RTL code using LLMs and symbolic reasoning techniques, while~\cite{learntodebug,hlsdebugger,assertllm} use LLMs for RTL debugging and verification.
These works demonstrate the potential of LLMs across RTL generation, optimization, debugging, and verification.
They primarily treat RTL as the generation target, whereas \CPPL studies a complementary question: how to expose compiler-level hardware design flows to LLMs without requiring the model to directly emit low-level CIRCT IR.

\subsection{Circuit IR Compilers and Tools}
\label{sec:circt}

Circuit IR Compilers and Tools (CIRCT)~\cite{circt} is a compiler infrastructure built on top of the Multi-Level Intermediate Representation (MLIR)~\cite{mlir} for hardware design~\cite{assassyn, cement, chisel, calyx}, optimization~\cite{pipertl, combrewriter}, and simulation~\cite{llhd, Khronos}.
CIRCT provides core dialects to represent circuits at a unified level of abstraction, including:
\begin{itemize}[leftmargin=*, itemindent=0pt, listparindent=0pt]
	\item The \texttt{hw} dialect offers function-like semantics to represent module information and data types. For instance, \texttt{hw.module} handles the details of a module, while \texttt{hw.instance} represents the instantiation of these modules.
	\item The \texttt{comb} dialect represents combinational components in RTL. For example, \texttt{comb.add} models a multi-input adder in combinational logic.
	\item The \texttt{seq} dialect represents sequential logic. \texttt{seq.compreg} models registers in CIRCT, containing the piped value and reset as inputs. The \texttt{seq} dialect also includes a memory type that describes memory behavior.
	\item The \texttt{sv} dialect represents the semantics of SystemVerilog. For example, \texttt{sv.always} represents an always block, which is commonly used to define sequential logic.
\end{itemize}
The core dialects of CIRCT effectively support flexible transformations and optimization, while targeting different backends for Verilog generation, simulation, and verification.

\subsection{Motivations for CPPL}
\label{sec:motivation}

The previous works on LLMs for RTL generation discussed in \Cref{sec:related-works} primarily follow an end-to-end paradigm, where the LLM directly generates the final RTL code.
This approach is convenient, but it exposes several limitations when used as a design methodology:
\begin{enumerate}[label=\protect\circled{\arabic*}, leftmargin=*, itemindent=0pt, listparindent=0pt]
    \item Generated RTL may contain syntax or tool-compatibility errors that are only discovered after downstream compilation.
    \item Functional behavior, bitwidth consistency, and structural hierarchy are entangled in a single text-generation task, making failures difficult to localize.
    \item Compiler analyses and transformations are applied only after RTL emission, so the generation process receives little benefit from typed IRs, legality checks, and canonical compiler optimizations.
\end{enumerate}

Compiler-based hardware design flows address many of these issues by making structure, legality, and types explicit before backend code generation.
For example, CIRCT can lower well-formed circuit IR into synthesizable Verilog for backend EDA tools.
It also provides standard optimization passes, including \textit{constant folding} (CF), \textit{dead code elimination} (DCE), and \textit{common subexpression elimination} (CSE), to improve the generated design.
Furthermore, CIRCT IR follows a strict type discipline derived from the LLVM/MLIR ecosystem~\cite{llvm}, enabling early verification before final RTL emission.

Based on the above discussion, we propose a compiler-mediated generation paradigm that combines the natural-language and code-generation capabilities of LLMs with the typed compilation flow provided by CIRCT.
We designate this paradigm as \CPPL (\underline{C}ircuit \underline{P}rompt \underline{P}rogramming \underline{L}anguage), a hardware generation framework that inserts an LLM-friendly, statically checkable circuit IR between model generation and CIRCT lowering.
\CPPL exposes structural and type constraints before backend code generation, while still allowing designers to express implementation intent at a high level.

\subsection{Challenges for Leveraging CIRCT IR in CPPL}
\label{sec:challenges}

Although CIRCT provides the desired compiler support, directly exposing CIRCT IR as the LLM output target presents several challenges:

\begin{enumerate}[label=\protect\circled{\arabic*}, leftmargin=*, itemindent=0pt, listparindent=0pt]
    \item There are limited public datasets of CIRCT IR examples, which makes it difficult to train LLMs to understand and generate code in this intermediate representation.
    \item CIRCT is an actively developing project, and its IR syntax and dialect definitions continue to evolve. This makes it difficult for LLMs to maintain compatibility with a specific compiler version.
    \item As discussed in~\cite{llm4ir}, the Static Single Assignment (SSA) form used by CIRCT IR can be difficult for LLMs to understand and generate.
\end{enumerate}

We conduct a profiling experiment to assess LLM performance in end-to-end generation of Verilog code versus CIRCT IR.
In this experiment, we use system prompts to instruct LLMs to generate CIRCT IR compatible with our experimental setup in \Cref{sec:exp-setup}.
We use the RTLLM benchmark~\cite{rtllm} to evaluate the \texttt{pass@1} metric on a set of recent commercial LLMs (see \Cref{sec:exp-setup} for configuration details).
As \Cref{fig:profiling} shows, all evaluated models exhibit a clear performance gap between generating Verilog code and CIRCT IR, with \texttt{pass@1} scores for CIRCT IR generation being substantially lower.
This result indicates that the direct compiler-IR path is not yet a reliable frontend for LLM-assisted hardware design.

\begin{figure}[tb!]
    \centering
    \includegraphics[width=0.8\linewidth]{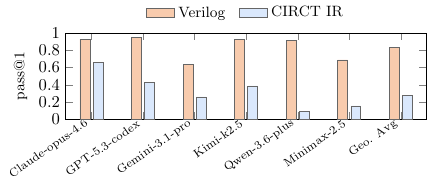}
    \caption{\texttt{pass@1} score gap between Verilog and CIRCT IR for syntax correctness.}
    \label{fig:profiling}
\end{figure}

\begin{figure}[tb!]
    \centering
    \includegraphics[width=0.95\linewidth]{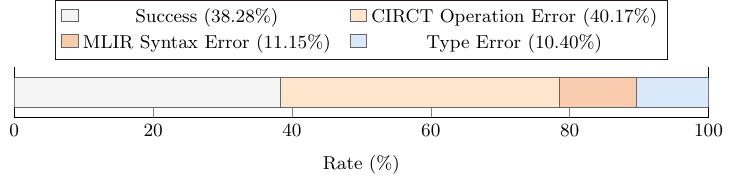}
    \caption{The geometric average error type breakdown of CIRCT IR generation across all evaluated models on the RTLLM benchmark.}
    \label{fig:breakdown}
\end{figure}

We further evaluate the detailed error types in generated CIRCT IR and categorize them into three main categories: violations of MLIR format semantics (\textit{MLIR Syntax Error}), use of unsupported or incorrect CIRCT operations (\textit{CIRCT Operation Error}), and type system mismatches (\textit{Type Error}).
As shown in \Cref{fig:breakdown}, the largest error category is CIRCT operation errors, indicating that LLMs often fail to satisfy dialect-specific operation constraints.
MLIR syntax errors and type errors also contribute to the overall error rate, confirming that both the concrete IR format and the compiler type system are challenging generation targets.

These challenges motivate an intermediate representation that keeps the compiler-visible structure and type information needed by CIRCT, but presents them in a form that is easier for LLMs to generate and easier for the compiler to validate.

\section{CPPL Framework}
\label{sec:cppl-framework}

\begin{figure}[!tb]
     \centering
    \includegraphics[width=.82\linewidth]{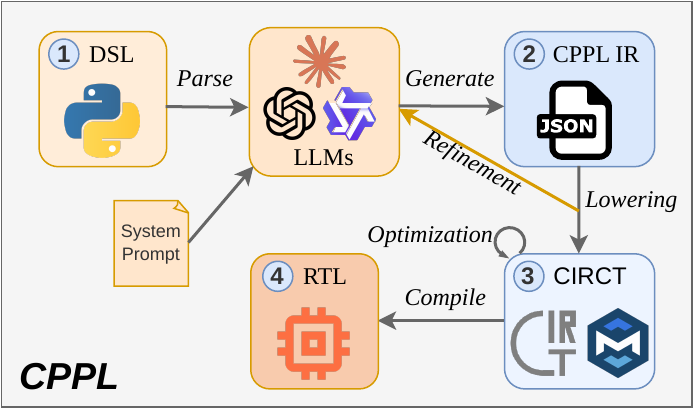}
    \caption{\textit{CPPL} framework overview.}
    \label{fig:cppl}
\end{figure}

In this section, we introduce the \CPPL framework, a compiler-mediated design flow for LLM-assisted hardware generation with CIRCT as the backend compilation infrastructure.
\Cref{fig:cppl} illustrates the overall workflow: \CPPL provides a Python-based frontend DSL for specifying interfaces and hierarchy, an LLM-facing JSON IR for behavioral generation, a compiler-driven refinement loop, and a deterministic lowering path to CIRCT.
\CPPL is implemented through APPL~\cite{appl} and CIRCT's Python bindings, which together provide a flexible platform for hardware design.

\begin{figure}[tb!]
\centering
\begin{minted}[numbersep=3pt, linenos, xleftmargin=7pt, breaklines,breakanywhere,frame=single,fontsize=\tiny,baselinestretch=0.92]{python}
from cppl import module, In, Out

@module
def Adder8(a: In[8], b: In[8]) -> {"sum": Out[8]}:
    """out equals a plus b (8-bit addition)."""

@module
def ALU(op_code: In[2], op_a: In[8], op_b: In[8]) -> {"res": Out[8], "zero": Out[1]}:
    return f"""
    Simple ALU that uses an Adder8 instance for addition.
    Based on op_code (2-bit selector):
    - 00: res = {Adder8(op_a, op_b)} (result from Adder8 instance)
    - 01: res = op_a - op_b
    - 10: res = op_a & op_b (bitwise AND)
    - 11: res = op_a | op_b (bitwise OR)
    zero is 1 when res equals 0, otherwise 0.
    """
\end{minted}

\caption{The \CPPL description of a 2-bit opcode ALU with an 8-bit adder instance.}
\label{fig:alu-cppl}
\end{figure}

\subsection{Frontend DSL for Hardware Design}
\label{sec:frontend-dsl}

\minisection{Language Design}
\Cref{fig:alu-cppl} shows a \CPPL program that defines an \texttt{ALU} module.
\CPPL follows a function-oriented style: each hardware module is written as a Python function decorated with the \texttt{@module} decorator, which triggers JIT compilation by the \CPPL compiler.
The function signature specifies the module boundary.
Arguments annotated with \texttt{In} define input ports, and return values annotated with \texttt{Out} define output ports, together with their bit widths.
Because the interface is declared explicitly, generated hardware is constrained to use fixed and well-typed ports, reducing the opportunity for the LLM to introduce missing, extra, or width-mismatched interfaces.

The function body provides the implementation intent through a docstring.
This docstring can mix natural-language descriptions with structured \CPPL constructs that describe module instantiations and connections.
For example, in \Cref{fig:alu-cppl}, the \texttt{ALU} module instantiates an \texttt{Adder8} module for operand addition.
The instantiation is expressed using Python string-formatting syntax, where a formatted module call explicitly binds the child module to ports declared in the enclosing function signature or in the local implementation context.
In this way, \CPPL keeps the high-level specification natural for designers while making the module hierarchy and interface bindings explicit enough for deterministic frontend elaboration.

\minisection{Structural Semantics Preservation}
We next formalize the structural part of the frontend DSL and show that the hierarchy encoded by \CPPL is preserved in the generated circuit representation.
The key point is that string-formatting constructs are not treated as unconstrained free-form text.
Instead, each formatted module call is interpreted as an explicit structural directive: it is lowered to a module-instance vertex, while the associated port bindings are lowered to wiring edges.
Therefore, the frontend does not rely on the LLM to rediscover hierarchy from natural language; it deterministically elaborates the hierarchy specified by the \CPPL syntax.
Let $s$ be a well-formed architectural phrase under scope environment $\Gamma$.
The judgment $\Gamma;\rho \vdash s \Downarrow H$ means that, under parent module or instance $\rho$, phrase $s$ elaborates to a circuit graph $H=(V,E_h,E_w)$, where $V$ is the set of module/instance vertices, $E_h$ is the set of hierarchy edges, and $E_w$ is the set of wiring edges.
Let $\mathsf{Struct}(H)=(V,E_h)$ denote the hierarchy-only projection, and let $\mathsf{Tree}_{\rho}(s)$ denote the module-instance parse tree obtained from the syntactic nesting of module declarations and instance calls in $s$ under root $\rho$.

The following core big-step rules capture the relevant structural behavior:
{\footnotesize
\[
\frac{}{\Gamma;\rho \vdash \texttt{module }m\{\} \Downarrow (\{m\},\{(\rho,m)\},\emptyset)}
\ \textsc{(S-Prim)}
\]
\[
\frac{\Gamma;\rho \vdash s \Downarrow (V,E_h,E_w)\quad M\in\Gamma}
{\Gamma;\rho \vdash \texttt{inst }x{:}M;\ s \Downarrow (V\cup\{x\},E_h\cup\{(\rho,x)\},E_w)}
\ \textsc{(S-Inst)}
\]
\[
\frac{\Gamma;\rho \vdash s \Downarrow (V,E_h,E_w)\quad p,q\in \mathsf{Ports}_{\Gamma}(V)}
{\Gamma;\rho \vdash \texttt{connect }p\to q;\ s \Downarrow (V,E_h,E_w\cup\{(p,q)\})}
\ \textsc{(S-Conn)}
\]
\[
\frac{\Gamma;\rho \vdash s_1 \Downarrow (V_1,E_{h1},E_{w1})\quad
      \Gamma;\rho \vdash s_2 \Downarrow (V_2,E_{h2},E_{w2})}
{\Gamma;\rho \vdash s_1;s_2 \Downarrow (V_1\cup V_2,E_{h1}\cup E_{h2},E_{w1}\cup E_{w2})}
\ \textsc{(S-Seq)}
\]
}

\begin{theorem}[Structural Semantics Preservation]
If $s$ is well formed and $\Gamma;\rho \vdash s \Downarrow H$, then $\mathsf{Struct}(H)$ is isomorphic to $\mathsf{Tree}_{\rho}(s)$.
\end{theorem}

\begin{proof}
By induction on the derivation of $\Gamma;\rho \vdash s \Downarrow H$.
\textsc{S-Prim} adds one module vertex and one hierarchy edge from $\rho$ to that module, matching the corresponding leaf in $\mathsf{Tree}_{\rho}(s)$.
\textsc{S-Inst} adds exactly one instance vertex and one hierarchy edge from the current parent $\rho$ to the instance $x$; by the induction hypothesis, the remaining phrase $s$ preserves its syntactic hierarchy, so the whole graph matches the parse tree after the instance expansion.
\textsc{S-Conn} only updates $E_w$, and thus leaves $\mathsf{Struct}(H)$ unchanged.
\textsc{S-Seq} follows from the induction hypotheses for $s_1$ and $s_2$ and from taking the union of their hierarchy vertices and edges under the same parent $\rho$.
Therefore, for every well-formed phrase, the hierarchy represented in the generated graph is the same as the hierarchy specified by the DSL syntax.
\end{proof}

\minisection{Code Generation}
\Cref{fig:alu-compile} illustrates the compilation flow from the \CPPL \texttt{ALU} module to Verilog.
The frontend first parses the module signatures, structured docstring constructs, and hierarchy implied by formatted module calls.
It then emits the corresponding CIRCT IR and invokes CIRCT's code-generation pipeline to produce Verilog.
This flow separates the LLM-facing specification from the backend representation: the LLM helps generate implementation details, while \CPPL and CIRCT enforce structural consistency and backend legality.

\begin{figure}[tb!]
\centering
\begin{minted}[numbersep=3pt, linenos, xleftmargin=7pt,breaklines,frame=single, highlightlines={9,11},fontsize=\tiny]{python}
from cppl import Design
design = Design() # Create a new design
design.add(ALU) # Add the top module to the design
design.to_verilog() # Compile the design to Verilog
\end{minted}
\caption{Code snippet using \CPPL to compile the ALU example into Verilog.}
\label{fig:alu-compile}
\end{figure}

\begin{figure}[tb!]
\centering
\begin{minted}[numbersep=3pt, linenos, xleftmargin=7pt,breaklines,breakanywhere, frame=single,fontsize=\tiny]{json}
[
  {
    "name": "Adder8",
    "ports": {
      "a":   { "dir": "input",  "width": 8 }, 
      "b":   { "dir": "input",  "width": 8 }, 
      "sum": { "dir": "output", "width": 8 }
    },
    "body": [
      { "id": "sum_val", "op": "add", "args": ["a", "b"] },
      { "op": "output", "args": { "sum": "sum_val" } }
    ]
  },
  {
    "name": "ALU",
    "ports": {
      "op_code": { "dir": "input",  "width": 2 }, 
      "op_a":    { "dir": "input",  "width": 8 }, 
      "op_b":    { "dir": "input",  "width": 8 }, 
      "res":     { "dir": "output", "width": 8 }, 
      "zero":    { "dir": "output", "width": 1 }
    },
    "body": [
      { "id": ["adder8_sum"], "op": "instance", "module": "Adder8", "args": { "a": "op_a", "b": "op_b" } },
      { "id": "sel0", "op": "extract", "args": ["op_code"], "lowBit": 0, "width": 1 },
      { "id": "sel1", "op": "extract", "args": ["op_code"], "lowBit": 1, "width": 1 },
      { "id": "sub_res", "op": "sub", "args": ["op_a", "op_b"] },
      { "id": "and_res", "op": "and", "args": ["op_a", "op_b"] },
      { "id": "or_res",  "op": "or",  "args": ["op_a", "op_b"] },
      { "id": "mux_lo", "op": "mux", "args": ["sel0", "sub_res", "adder8_sum"] },
      { "id": "mux_hi", "op": "mux", "args": ["sel0", "or_res", "and_res"] },
      { "id": "res_mux", "op": "mux", "args": ["sel1", "mux_hi", "mux_lo"] },
      { "id": "any_set", "op": "or_reduce", "args": ["res_mux"] },
      { "id": "is_zero", "op": "not", "args": ["any_set"] },
      { "op": "output", "args": { "res": "res_mux", "zero": "is_zero" } }
    ]
  }
]
\end{minted}
\caption{The \CPPLIR generated from the \CPPL code of the ALU example.}
\label{fig:alu-json}
\end{figure}

\subsection{JSON-based Intermediate Representation}
\label{sec:json-ir}

As discussed in \Cref{sec:challenges}, directly generating CIRCT IR poses challenges for LLMs.
However, LLMs' demonstrated ability to generate Verilog code suggests they possess some understanding of circuit structure and semantics. 
We leverage this capability by introducing an LLM-friendly intermediate representation that preserves compiler-visible structure while avoiding raw MLIR syntax.
JSON is a widely adopted data format with strong LLM support, as evidenced by its use in LLM serving and application frameworks~\cite{sglang, vllm, llamaindex}.
To this end, we design a JSON-based IR, \CPPLIR, that captures essential structural and behavioral circuit information in a regular, statically checkable format.

\minisection{Design Principles}
\CPPLIR is designed around three principles.
First, it separates the model-facing representation from the backend IR syntax, so LLMs do not need to emit MLIR punctuation, SSA names, or dialect assembly forms directly.
Second, it preserves compiler-relevant information: modules, ports, instances, operation identifiers, and operands are explicit fields rather than unstructured text.
Third, it defines a deterministic lowering contract to CIRCT, allowing the frontend to validate generated programs before Verilog emission.
This design enables early error detection and iterative refinement, while insulating the LLM from version-specific CIRCT syntax changes.
To reduce type mismatch errors, \CPPLIR is also designed with an inferable width system that can be statically checked during JSON-to-CIRCT compilation, as discussed in \Cref{sec:compile}.

\Cref{fig:alu-json} shows the \CPPLIR generated from the \CPPL description of the ALU example in \Cref{fig:alu-cppl}.
The function signature is translated into module declarations and port maps, corresponding to the \CPPL code structure.
Instance items are automatically inserted into the module body by the \CPPL compiler through parsing the docstring and formatted module calls prior to LLM generation, preserving module hierarchy and port bindings.
The remaining operations are then generated by the LLM based on the implementation intent described in the docstring.

\begin{figure}[!tb]
\centering
\begingroup
\scriptsize
\setlength{\arraycolsep}{3pt}
\[
\begin{array}{@{}llll@{}}
m\in\mathsf{ModName} & p\in\mathsf{PortName} &
x\in\mathsf{Id} & o\in\mathsf{Opcode}
\end{array}
\quad
w,n\in\mathbb{N}^{+}
\]
\vspace{-0.9em}
\[
\begin{array}{@{}rcl@{}}
D &::=& [M_1,\ldots,M_n]\\
M &::=& \{\texttt{"name"}:m,\ \texttt{"ports"}:P,\ \texttt{"body"}:B\}\\
P &::=& \{p_i:\Pi_i\}_{i=1}^{n}\\
\Pi &::=& \{\texttt{"dir"}:\delta,\ \texttt{"width"}:w\}\\
\delta &::=& \texttt{"input"} \mid \texttt{"output"}\\
B &::=& [O_1,\ldots,O_n]\\[0.15em]
O &::=& \{\texttt{"id"}:r,\ \texttt{"op"}:o,\ \texttt{"args"}:A,\ \kappa\}\\
  &\mid& \{\texttt{"id"}:[x_1,\ldots,x_n],\ \texttt{"op"}:\texttt{"instance"},\\
  &&\quad \texttt{"module"}:m,\ \texttt{"args"}:\Gamma,\ \kappa\}\\
  &\mid& \{\texttt{"op"}:\texttt{"output"},\ \texttt{"args"}:\Gamma\}\\[0.15em]
r &::=& x \mid [x_1,\ldots,x_n]\\
A &::=& [x_1,\ldots,x_n] \mid \Gamma\\
\Gamma &::=& \{p_i:x_i\}_{i=1}^{n}\\
\kappa &::=& \epsilon \mid a:v,\kappa
\end{array}
\]
\endgroup
\caption{Formal syntax of the JSON-based IR.}
\label{fig:json-ir-syntax}
\end{figure}

\minisection{IR Semantics}
We formalize the syntax of the \CPPLIR in \Cref{fig:json-ir-syntax}.
The grammar follows the concrete JSON structure: a design is a list of modules, each module contains a port map and a body of dictionary-based operations.
The body of each module consists of operations and module instances, each identified by a unique identifier except for output operations.
This identifier scheme enforces a definition-before-use discipline, maintaining consistency with CIRCT IR semantics.
The \CPPL compiler can also insert predefined items, such as module instances extracted from the frontend DSL, into the module body before LLM-generated operations are compiled.
The syntax rules serve as the system prompt to guide LLMs in generating \CPPLIR that can be correctly parsed and compiled to CIRCT IR.
The \CPPLIR used in the \CPPL framework bridges LLM generation capabilities with CIRCT's static compilation flow, enabling more reliable LLM-assisted hardware design.

\begin{figure*}[t]
\centering
\begingroup
\scriptsize
\setlength{\arraycolsep}{3pt}
\setlength{\tabcolsep}{0pt}
\renewcommand{\arraystretch}{1.8}
\newcommand{\typerule}[3]{%
\begin{minipage}[t]{0.32\textwidth}
\centering
\vspace{0pt}
$\dfrac{#1}{\begin{gathered}#2\end{gathered}}\ \textsc{#3}$
\end{minipage}}
\begin{tabular}{@{}ccc@{}}
\typerule{}
{\Sigma;\Gamma \vdash \mathsf{const}(x,w)\\
 \Rightarrow \Gamma[x\mapsto i^w]}
{T-Const}
&
\typerule{\Gamma(a)=i^w \quad u\in\mathsf{Unary}}
{\Sigma;\Gamma \vdash \mathsf{unary}_{u}(x,a)\\
 \Rightarrow \Gamma[x\mapsto i^w]}
{T-Unary}
&
\typerule{\Gamma(a)=i^w \quad r\in\mathsf{Reduce}}
{\Sigma;\Gamma \vdash \mathsf{reduce}_{r}(x,a)\\
 \Rightarrow \Gamma[x\mapsto i^1]}
{T-Reduce}
\\
\typerule{\Gamma(a_1)=i^w \quad \Gamma(a_2)=i^w \quad b\in\mathsf{Binary}}
{\Sigma;\Gamma \vdash \mathsf{bin}_{b}(x,a_1,a_2)\\
 \Rightarrow \Gamma[x\mapsto i^w]}
{T-Bin}
&
\typerule{\Gamma(a_1)=i^w \quad \Gamma(a_2)=i^w \quad c\in\mathsf{Cmp}}
{\Sigma;\Gamma \vdash \mathsf{cmp}_{c}(x,a_1,a_2)\\
 \Rightarrow \Gamma[x\mapsto i^1]}
{T-Cmp}
&
\typerule{\Gamma(s)=i^1 \quad \Gamma(t)=i^w \quad \Gamma(f)=i^w}
{\Sigma;\Gamma \vdash \mathsf{mux}(x,s,t,f)\\
 \Rightarrow \Gamma[x\mapsto i^w]}
{T-Mux}
\\
\typerule{\Gamma(a)=i^{w_s} \quad w_s \leq w}
{\Sigma;\Gamma \vdash \mathsf{cast}(x,a,w)\\
 \Rightarrow \Gamma[x\mapsto i^w]}
{T-Cast}
&
\typerule{\Gamma(a_i)=i^{w_i}\ (1\leq i\leq n)}
{\Sigma;\Gamma \vdash \mathsf{concat}(x,\vec{a})\\
 \Rightarrow \Gamma[x\mapsto i^{\sum_{i=1}^{n}w_i}]}
{T-Concat}
&
\typerule{\Gamma(a)=i^{w_s} \quad l+w\leq w_s}
{\Sigma;\Gamma \vdash \mathsf{extract}(x,a,l,w)\\
 \Rightarrow \Gamma[x\mapsto i^w]}
{T-Extract}
\\
\typerule{\Gamma(d)=i^w \quad \Gamma(clk)=i^1 \quad \Gamma(en)=i^1}
{\Sigma;\Gamma \vdash \mathsf{reg}(x,d,clk,en,w)\\
 \Rightarrow \Gamma[x\mapsto i^w]}
{T-Reg}
&
\typerule{\begin{gathered}
  \Sigma(m)=(\{p_i:i^{w_i}\}_{i=1}^{n},\{q_j:i^{v_j}\}_{j=1}^{k})\\
  \rho(p_i)=a_i \quad \Gamma(a_i)=i^{w_i}\ (1\leq i\leq n)
  \end{gathered}}
{\Sigma;\Gamma \vdash \mathsf{inst}(\vec{x},m,\rho)\\
 \Rightarrow \Gamma[x_j\mapsto i^{v_j}]_{j=1}^{k}}
{T-Inst}
&
\typerule{O(p_i)=i^{w_i} \quad \rho(p_i)=a_i \quad \Gamma(a_i)=i^{w_i}}
{\Sigma;\Gamma \vdash \mathsf{out}(\rho,O)\\
 \Rightarrow \Gamma}
{T-Out}
\end{tabular}
\endgroup
\caption{Representative width inference rules for \CPPLIR.}
\label{fig:type-infer}
\vspace{-0.8em}
\end{figure*}

\subsection{Backend Compilation}
\label{sec:compile}
The \CPPL compiler lowers \CPPLIR into CIRCT IR and invokes CIRCT's code-generation pipeline to emit Verilog.
During JSON-to-CIRCT compilation, two steps are central to reliability: \textit{type inference} and \textit{code refinement}.

\minisection{Type Inference}
\CPPLIR uses bitwidths as its type system, aligning with CIRCT IR's type system where \texttt{hw.integer} types represent unsigned integers of specific widths.
For a module $m$, its signature environment $\Sigma(m)=(I,O)$ records the declared input and output port types.
Starting from the input-port environment $\Gamma_0=I$, the compiler infers each internal SSA value by propagating widths through operations.
We write $\Sigma;\Gamma \vdash O \Rightarrow \Gamma'$ to denote that operation $O$ is well typed under $\Gamma$ and extends it to $\Gamma'$.
\Cref{fig:type-infer} summarizes the representative rules: arithmetic and logic operations preserve width, comparisons and reductions produce one-bit values, muxes require equal-width branches, and instances obtain result widths from callee signatures.
The implementation applies these rules to a fixpoint and then checks the inferred output values against $O$.
This rejects width errors such as unequal binary operands, non-1-bit mux selectors, invalid extracts, and output bindings whose inferred widths do not match the declared ports.

\begin{figure}[ht]
\centering
\begin{minted}[numbersep=3pt, linenos, xleftmargin=7pt,breaklines,breakanywhere, frame=single,fontsize=\tiny]{mlir}
module {
  hw.module @Adder8(in %a : i8, in %b : i8, out sum : i8) {
    %0 = comb.add %a, %b : i8
    hw.output %0 : i8
  }
  hw.module @ALU(in %op_code : i2, in %op_a : i8, in %op_b : i8, out res : i8, out zero : i1) {
    %Adder8_0.sum = hw.instance "Adder8_0" @Adder8(a: %op_a: i8, b: %op_b: i8) -> (sum: i8)
    %0 = comb.extract %op_code from 0 : (i2) -> i1
    %1 = comb.extract %op_code from 1 : (i2) -> i1
    %2 = comb.sub %op_a, %op_b : i8
    %3 = comb.and %op_a, %op_b : i8
    %4 = comb.or %op_a, %op_b : i8
    %5 = comb.mux %0, %2, %Adder8_0.sum : i8
    %6 = comb.mux %0, %4, %3 : i8
    %7 = comb.mux %1, %6, %5 : i8
    %c0_i8 = hw.constant 0 : i8
    %8 = comb.icmp ne %7, %c0_i8 : i8
    %true = hw.constant true
    %9 = comb.xor %8, %true : i1
    hw.output %7, %9 : i8, i1
  }
}
\end{minted}

\caption{The generated CIRCT IR of the ALU example compiled from the JSON-based IR in \Cref{fig:alu-json}.}
\label{fig:alu-mlir}
\end{figure}

\minisection{Code Refinement}
In addition to reporting type errors, the \CPPL compiler performs static analyses on \CPPLIR to identify potential issues before CIRCT lowering.
The core analyses are listed as follows:
\begin{itemize}[leftmargin=*, itemindent=0pt, listparindent=0pt]
\item \textbf{Syntax validation}: The \CPPL compiler validates the syntax of the generated \CPPLIR to ensure it conforms to the defined grammar and structure. This includes verifying identifier symbol tables, terminator coverage, and instance graph connectivity.
\item \textbf{Dead code elimination}: The \CPPL compiler identifies and eliminates unused operations or module instances in the final design, reducing unnecessary IR before backend compilation.
\item \textbf{Combinational loop detection}: The \CPPL compiler detects combinational loops in the generated \CPPLIR, which would result in non-synthesizable designs.
\end{itemize}
Each identified issue is wrapped in a Python Exception with a descriptive error message, which is used to provide feedback to the LLM for iterative refinement of the generated \CPPLIR.
This refinement loop is shown in \Cref{fig:cppl}.
The maximum number of iterations is denoted as $N_{\max}$, and the iteration process continues until a valid \CPPLIR is generated or the maximum number of iterations is reached.

\begin{figure}[ht]
\centering
\begin{minted}[numbersep=3pt,breaklines,breakanywhere,linenos,xleftmargin=7pt,frame=single,fontsize=\tiny,baselinestretch=0.92]{verilog}
module ALU(
  input  [1:0] op_code,
  input  [7:0] op_a,
               op_b,
  output [7:0] res,
  output       zero
);

  wire [7:0] _Adder8_0_sum;
  Adder8 Adder8_0 (
    .a   (op_a),
    .b   (op_b),
    .sum (_Adder8_0_sum)
  );
  wire [7:0] _GEN =
    op_code[1]
      ? (op_code[0] ? op_a | op_b : op_a & op_b)
      : op_code[0] ? op_a - op_b : _Adder8_0_sum;
  assign res = _GEN;
  assign zero = ~(|_GEN);
endmodule
\end{minted}

\caption{The generated Verilog code of the ALU example compiled from the CIRCT IR in \Cref{fig:alu-mlir}.}
\label{fig:alu-verilog}
\end{figure}

\subsection{Code Generation}
\Cref{fig:alu-mlir} shows the CIRCT IR generated for the ALU example, compiled from the \CPPLIR in \Cref{fig:alu-json}.
Module declarations and port mappings are directly translated from the \CPPLIR, while operations are generated by the LLM based on the implementation intent described in the \CPPL code's docstring.
The generated CIRCT IR is then passed through CIRCT's code generation pipeline to emit the final Verilog code shown in \Cref{fig:alu-verilog}.
The resulting Verilog code is syntactically valid and accepted by backend EDA tools in our evaluation.

\Cref{fig:pipeline} illustrates the compilation pipeline for Verilog code generation. The highlighted lines show the optimization passes applied before Verilog emission, which include the standard optimizations discussed in \Cref{sec:motivation}.
CIRCT's code generation pipeline applies standard canonicalization and optimization passes, thereby improving the quality of the emitted Verilog without requiring the LLM to implement these transformations explicitly.

Overall, \CPPL provides a structured and modular approach to LLM-assisted hardware generation, combining LLM generation with compiler-checked legality and backend optimization.

\begin{figure}[ht]
\centering
\begin{minted}[numbersep=3pt, linenos, xleftmargin=7pt, breaklines,breakanywhere,frame=single,fontsize=\tiny, baselinestretch=0.92, highlightlines={4, 5}]{python}
pipeline = [
    "hw.module(lower-seq-hlmem)", 
    "lower-seq-to-sv",
    "canonicalize",
    "cse",
    "hw.module(prettify-verilog)",
    "hw.module(hw-cleanup)",
]
pm = passmanager.PassManager.parse("builtin.module(" + ",".join(pipeline) + ")")
pm.run(mod.operation)
\end{minted}

\caption{The CIRCT pass pipeline used in the \CPPL compiler to generate Verilog code from the generated CIRCT IR.}
\label{fig:pipeline}
\end{figure}

\section{Evaluation}
\label{sec:experiments}

\begin{table*}[!t]
\centering
\scriptsize
\setlength{\tabcolsep}{2pt}
\begin{minipage}[t]{0.49\textwidth}
\centering
\caption{Syntax and functional correctness of direct SystemVerilog generation on RTLLM.}
\label{tab:pass-sv}
\renewcommand{\arraystretch}{1.05}
\resizebox{0.9\linewidth}{!}{%
\begin{tabular}{lcccccc}
\hline
\multirow{2}{*}{Model}
& \multicolumn{3}{c}{Syntax}
& \multicolumn{3}{c}{Functionality} \\
\cmidrule(lr){2-4}\cmidrule(lr){5-7}
& pass@1 & pass@2 & pass@5 & pass@1 & pass@2 & pass@5 \\
\hline
Claude-opus-4.6 & 0.932 & 0.932 & 0.956 & 0.725 & 0.761 & 0.778 \\
GPT-5.3-codex   & 0.957 & 0.963 & 0.964 & 0.675 & 0.710 & 0.752 \\
Gemini-3.1-pro  & 0.639 & 0.709 & 0.759 & 0.554 & 0.603 & 0.642 \\
Kimi-k2.5       & 0.925 & 0.963 & 0.964 & 0.704 & 0.754 & 0.782 \\
Qwen-3.6-plus   & 0.918 & 0.943 & 0.961 & 0.679 & 0.737 & 0.798 \\
Minimax-2.5     & 0.689 & 0.799 & 0.912 & 0.543 & 0.631 & 0.748 \\
\hline
\end{tabular}%
}
\end{minipage}\hfill
\begin{minipage}[t]{0.49\textwidth}
\centering
\caption{Syntax and functional correctness of direct CIRCT IR generation on RTLLM.}
\label{tab:pass-circt}
\renewcommand{\arraystretch}{1.05}
\resizebox{0.9\linewidth}{!}{%
\begin{tabular}{lcccccc}
\hline
\multirow{2}{*}{Model}
& \multicolumn{3}{c}{Syntax}
& \multicolumn{3}{c}{Functionality} \\
\cmidrule(lr){2-4}\cmidrule(lr){5-7}
& pass@1 & pass@2 & pass@5 & pass@1 & pass@2 & pass@5 \\
\hline
Claude-opus-4.6 & 0.661 & 0.790 & 0.884 & 0.500 & 0.598 & 0.693 \\
GPT-5.3-codex   & 0.425 & 0.542 & 0.679 & 0.318 & 0.394 & 0.471 \\
Gemini-3.1-pro  & 0.257 & 0.316 & 0.381 & 0.239 & 0.285 & 0.328 \\
Kimi-k2.5       & 0.386 & 0.476 & 0.580 & 0.264 & 0.325 & 0.384 \\
Qwen-3.6-plus   & 0.093 & 0.137 & 0.208 & 0.082 & 0.117 & 0.168 \\
Minimax-2.5     & 0.154 & 0.233 & 0.403 & 0.089 & 0.113 & 0.171 \\
\hline
\end{tabular}%
}
\end{minipage}
\end{table*}

\begin{table}[!t]
\centering
\caption{Functional correctness of generated Verilog code from CPPL implementations on different models.}
\label{tab:cppl-results}
\scriptsize
\setlength{\tabcolsep}{4pt}
\renewcommand{\arraystretch}{1.05}
\resizebox{0.7\linewidth}{!}{%
\begin{tabular}{lccc}
\hline
\multirow{2}{*}{Model}
& \multicolumn{3}{c}{Functionality} \\
\cmidrule(lr){2-4}
& pass@1 & pass@2 & pass@5 \\
\hline
CPPL (Claude-opus-4.6)    & 0.800 & 0.817 & 0.838 \\
CPPL (GPT-5.3-codex) & 0.782 & 0.840 & 0.874 \\
CPPL (Qwen-3.6-plus)      & 0.768 & 0.812 & 0.821 \\
\hline
\end{tabular}%
}
\end{table}

\begin{table}[!t]
\centering
\caption{Average AIG node counts of CPPL-generated Verilog (Claude-opus-4.6) vs. RTLLM reference designs, with and without CIRCT optimizations (\Cref{fig:pipeline}).}
\label{tab:aig-results}
\scriptsize
\setlength{\tabcolsep}{2pt}
\renewcommand{\arraystretch}{1.05}
\newcommand{\refcell}[1]{\cellcolor{blue!8}#1}
\newcommand{\best}[1]{\cellcolor{green!25}\textbf{#1}}
\newcommand{\near}[1]{\cellcolor{yellow!25}#1}
\newcommand{\midcell}[1]{\cellcolor{CUHKorange!18}#1}
\newcommand{\bad}[1]{\cellcolor{red!18}#1}
\newcommand{\missing}{\cellcolor{black!12}--}
\newcommand{\legendbox}[1]{\begingroup\setlength{\fboxsep}{0pt}\fcolorbox{black!25}{#1}{\rule{0pt}{0.9ex}\rule{0.9ex}{0pt}}\endgroup}
\resizebox{0.6\linewidth}{!}{%
\begin{tabular}{lrrrr}
\hline
Design & Ref & Verilog & w/o Opt & Opt \\
\hline
ram              & \refcell{700}   & \bad{4147}       & \near{458.67}    & \best{413}    \\
accu             & \refcell{318}   & \near{309}       & \best{290}       & \best{290}    \\
adder\_16bit     & \refcell{608}   & \near{368}       & \midcell{386}    & \best{352}    \\
adder\_32bit     & \refcell{1982}  & \near{1038.1}    & \midcell{1045}   & \best{784}    \\
adder\_8bit      & \refcell{160}   & \midcell{200}        & \near{184}       & \best{168}    \\
adder\_pipe\_64bit & \refcell{1892} & \best{1892}      & \near{1909}      & \best{1892}   \\
alu              & \refcell{10159} & \best{6571.75}   & \midcell{8642.67}    & \near{8413}   \\
calendar         & \refcell{352}   & \bad{352}        & \near{295.78}    & \best{229}    \\
div\_16bit       & \refcell{16397} & \bad{6603}       & \bad{5945}       & \best{3945}   \\
edge\_detect     & \refcell{14}    & \best{7}         & \best{7}         & \best{7}      \\
freq\_div        & \refcell{207}   & \midcell{170.1}  & \near{140}       & \best{125}    \\
fsm              & \refcell{130}   & \bad{91.89}      & \bad{279}        & \best{58.12}  \\
multi\_16bit     & \refcell{2409}  & \near{2644.9}    & \best{2146}      & \best{2146}   \\
multi\_pipe\_4bit & \refcell{345}  & \near{345}       & \best{333}       & \best{333}    \\
multi\_pipe\_8bit & \refcell{1936} & \midcell{1877}   & \near{1840.71}   & \best{1834.6} \\
parallel2serial  & \refcell{63}    & \near{63}        & \best{49}        & \best{49}     \\
radix2\_div      & \refcell{941}   & \missing         & \best{1863}      & \best{1863}   \\
serial2parallel  & \refcell{127}   & \missing         & \best{81}        & \best{81}     \\
signal\_generator & \refcell{154}  & \missing         & \bad{221}        & \best{125}    \\
traffic\_light   & \refcell{493}   & \missing         & \near{437}       & \best{425.4}  \\
\hline
Geo. Avg         & \refcell{513.89} & \midcell{525.09} & \near{439.88}    & \best{368.16} \\
\hline
\end{tabular}%
}
\parbox{\linewidth}{\scriptsize\emph{Note:} Lower is better. \legendbox{blue!8} RTLLM reference; among generated results, \legendbox{green!25} best in each row, \legendbox{yellow!25}/\legendbox{CUHKorange!18} increasingly larger node counts, and \legendbox{red!18} more than $1.5\times$ the row best; \legendbox{black!12} unavailable.}
\end{table}

\subsection{Experiment Setup}
\label{sec:exp-setup}

We evaluate the performance of \CPPL on the RTLLM benchmark~\cite{rtllm}, which consists of 29 problems covering a wide range of hardware design tasks, including combinational logic, sequential logic, and memory design.
The LLMs evaluated in our experiments include Claude-opus-4.6, GPT-5.3-codex, Gemini-3.1-pro, Kimi-k2.5, Qwen-3.6-plus, and Minimax-2.5.
These models are selected based on their strong performance in code generation tasks and their wide adoption in recent research.

We employ the \texttt{pass@1} metric~\cite{verilogeval, codeeval} used in LLM code generation evaluation tasks to assess the performance of LLMs in generating correct CIRCT IR and Verilog code, which is defined as follows:
\begin{equation}
    pass@k := \mathbbm{E}_{\text{problems}} \left[ 1 - \frac{\binom{n-c}{k}}{\binom{n}{k}} \right],
\end{equation}
where $n\ge k$ samples are generated per problem and a problem is solved if any of the $k$ samples passes the unit tests.
In our experiments, we sample $n = 10$ code completions per problem for each downstream task and measure \texttt{pass@k} with $k = 1, 2, 5$.
For all evaluation tasks, we set the model temperature to $0.1$ to balance diversity and correctness in generated code.
The maximum token length is set to $4096$, covering reasoning traces, regenerations, and final code generation.
For the \CPPL refinement loop, we set the maximum number of refinement attempts to $N_{\max}=3$.

We evaluate the main stages of the \CPPL design flow, including syntax correctness, functionality verification, and backend performance evaluation.
Syntax correctness and functionality verification are assessed using \texttt{iverilog}~\cite{iverilog} with the simulation scripts and unit tests from the RTLLM benchmark.
Backend performance evaluation is conducted using the \texttt{Yosys} synthesis tool~\cite{yosys} to measure resource utilization of the generated RTL code.
The prompts used for generating CIRCT IR, Verilog code, and \CPPL implementations share the same design description, differing only in the instruction specifying the target output format.
For CIRCT IR generation, we include additional system prompts instructing the LLM to adhere to the CIRCT IR format and semantics compatible with our experimental framework.
We use direct Verilog and direct CIRCT IR generation as the primary baselines to isolate the effect of replacing unconstrained text generation with a compiler-checkable frontend representation.
Existing LLM4RTL systems~\cite{origen,betterv,rtlcoder,rtlpp,mage,spec2rtl,rtlrewriter} primarily target final RTL and often incorporate task-specific prompting, fine-tuning, retrieval, or repair strategies; they do not evaluate LLM generation of CIRCT IR or expose a comparable compiler-IR path.
Therefore, a direct end-to-end comparison with such systems would conflate our representation and compiler design with orthogonal model- and workflow-level techniques.
\CPPL is complementary to these techniques: the same prompting, retrieval, or repair strategies could be used to generate \CPPLIR instead of raw RTL.

\subsection{Performance on Verilog and CIRCT IR Generation}
\label{sec:exp-verilog-circt}
\Cref{tab:pass-sv} and \Cref{tab:pass-circt} present the syntax and functional correctness of direct SystemVerilog and CIRCT IR generation on the RTLLM benchmark, respectively.

Direct Verilog generation achieves high syntax correctness but still leaves a large gap to functional correctness.
Four models exceed 0.9 syntax \texttt{pass@1}, with GPT-5.3-codex reaching 0.957, yet the best functional \texttt{pass@1} is only 0.725 from Claude-opus-4.6.
This indicates that syntactically valid RTL does not necessarily implement the intended behavior, and that syntax checking alone is insufficient as a design-flow interface.

Direct CIRCT IR generation is substantially harder.
Even with format-specific system prompts, the syntax \texttt{pass@1} drops sharply for most models; for example, Qwen-3.6-plus and Minimax-2.5 achieve only 0.093 and 0.154, respectively.
Functional correctness is also limited, with all models except Claude-opus-4.6 staying below 0.5.
These results motivate an intermediate frontend representation that preserves the benefits of compiler-based hardware generation without requiring LLMs to emit raw CIRCT IR.

\subsection{Performance of CPPL}
\label{sec:exp-cppl}

To evaluate the effectiveness of \CPPL, we conduct experiments on selected models that span different direct CIRCT IR generation capabilities: Claude-opus-4.6, GPT-5.3-codex, and Qwen-3.6-plus.
\Cref{tab:cppl-results} presents the functional correctness of generated Verilog code from \CPPL implementations on different models.
All evaluated models produce syntactically valid CIRCT IR and Verilog through the \CPPL compilation flow for all RTLLM designs in our runs.
This result indicates that the JSON-based \CPPLIR and compiler checks avoid a major source of raw CIRCT IR generation failures.
The generated Verilog also achieves higher functional correctness than both direct Verilog generation and direct CIRCT IR generation across the evaluated models.
Notably, Qwen-3.6-plus, which performs poorly in direct CIRCT IR generation, reaches a \texttt{pass@1} score of 0.768 with \CPPL, compared with 0.082 for direct CIRCT IR functional correctness.
These results show that \CPPL can expose compiler-backed hardware generation to models that are otherwise unreliable at emitting raw CIRCT IR.

\subsection{Synthesis Node Reduction}
\label{sec:exp-synthesis}
To evaluate the backend performance of \CPPL-generated Verilog code, we conduct synthesis experiments using Yosys to measure AIG node counts.
Each design is processed by \texttt{read\_verilog -sv}, \texttt{hierarchy -top}, \texttt{synth -top -noabc}, \texttt{aigmap}, and \texttt{stat -json}; we report the post-\texttt{aigmap} cell count.
Similar synthesis-level proxies have also been used in prior works~\cite{betterv,symrtlo,logicnn}.
To isolate the impact of compiler optimizations, we perform an ablation study comparing designs generated with and without CIRCT optimization passes in the compilation pipeline.
As shown in \Cref{tab:aig-results}, \CPPL-generated Verilog often yields lower post-\texttt{aigmap} node counts than direct LLM-generated Verilog and the RTLLM reference designs.
The ablation further shows that enabling CIRCT optimization passes reduces the geometric average node count from 439.88 to 368.16, a 16.3\% reduction.
These results support the central design choice of \CPPL: once LLM output is represented as a compiler-checkable circuit IR, conventional hardware compiler optimizations can improve the emitted RTL without requiring the model to implement those transformations directly.
\FloatBarrier

\section{Conclusion}
\label{sec:conclu}

In this paper, we presented \CPPL, a compiler-mediated framework for LLM-assisted hardware generation.
\CPPL combines a Python frontend DSL, a statically checkable JSON-based circuit IR, and deterministic lowering to CIRCT, allowing LLMs to generate hardware through a structured representation rather than raw RTL or raw compiler IR.
By recovering widths from module ports, validating generated operations, and applying CIRCT backend optimizations, \CPPL exposes compiler checks and transformations to the LLM generation flow.
Our evaluation on RTLLM shows that this design improves functional correctness over direct Verilog and CIRCT IR generation and reduces post-synthesis AIG node counts through compiler optimization.
These results point to compiler-mediated generation as a promising direction for reliable LLM hardware design that can benefit from backend optimization.

{
    \bibliographystyle{IEEEtran}
    \bibliography{refs/Top-sim,refs/ref}

\begin{thebibliography}{10}
\providecommand{\url}[1]{#1}
\csname url@samestyle\endcsname
\providecommand{\newblock}{\relax}
\providecommand{\bibinfo}[2]{#2}
\providecommand{\BIBentrySTDinterwordspacing}{\spaceskip=0pt\relax}
\providecommand{\BIBentryALTinterwordstretchfactor}{4}
\providecommand{\BIBentryALTinterwordspacing}{\spaceskip=\fontdimen2\font plus
\BIBentryALTinterwordstretchfactor\fontdimen3\font minus
  \fontdimen4\font\relax}
\providecommand{\BIBforeignlanguage}[2]{{%
\expandafter\ifx\csname l@#1\endcsname\relax
\typeout{** WARNING: IEEEtran.bst: No hyphenation pattern has been}%
\typeout{** loaded for the language `#1'. Using the pattern for}%
\typeout{** the default language instead.}%
\else
\language=\csname l@#1\endcsname
\fi
#2}}
\providecommand{\BIBdecl}{\relax}
\BIBdecl

\bibitem{qwen3}
A.~Yang, A.~Li, B.~Yang, B.~Zhang, B.~Hui, B.~Zheng, B.~Yu, C.~Gao, C.~Huang,
  C.~Lv \emph{et~al.}, ``{Qwen3 technical report},'' \emph{arXiv preprint},
  2025.

\bibitem{deepseek}
A.~Liu, B.~Feng, B.~Xue, B.~Wang, B.~Wu, C.~Lu, C.~Zhao, C.~Deng, C.~Zhang,
  C.~Ruan \emph{et~al.}, ``{Deepseek-v3 technical report},'' \emph{arXiv
  preprint}, 2024.

\bibitem{gpt}
J.~Achiam, S.~Adler, S.~Agarwal, L.~Ahmad, I.~Akkaya, F.~L. Aleman, D.~Almeida,
  J.~Altenschmidt, S.~Altman, S.~Anadkat \emph{et~al.}, ``{Gpt-4 technical
  report},'' \emph{arXiv preprint}, 2023.

\bibitem{verilogeval}
M.~Liu, N.~Pinckney, B.~Khailany, and H.~Ren, ``{VerilogEval: Evaluating Large
  Language Models for Verilog Code Generation},'' in \emph{Proc.~ICCAD}, 2023.

\bibitem{rtllm}
Y.~Lu, S.~Liu, Q.~Zhang, and Z.~Xie, ``{RTLLM: An Open-source Benchmark for
  Designing RTL Generation with Large Language Model},'' in
  \emph{Proc.~ASPDAC}, 2024.

\bibitem{origen}
F.~Cui, C.~Yin, K.~Zhou, Y.~Xiao, G.~Sun, Q.~Xu, Q.~Guo, Y.~Liang, X.~Zhang,
  D.~Song \emph{et~al.}, ``{Origen: Enhancing RTL Code Generation with
  Code-to-Code Augmentation and Self-Reflection},'' in \emph{Proc.~ICCAD},
  2024.

\bibitem{betterv}
Z.~Pei, H.-L. Zhen, M.~Yuan, Y.~Huang, and B.~Yu, ``{BetterV: Controlled
  Verilog Generation with Discriminative Guidance},'' \emph{Proc.~ICML}, 2024.

\bibitem{rtlcoder}
S.~Liu, W.~Fang, Y.~Lu, J.~Wang, Q.~Zhang, H.~Zhang, and Z.~Xie, ``{RTLCoder:
  Fully Open-Source and Efficient LLM-Assisted RTL Code Generation
  Technique},'' \emph{IEEE TCAD}, 2024.

\bibitem{circt}
``{CIRCT: Circuit IR Compilers and Tools},'' \url{https://circt.llvm.org/},
  2026.

\bibitem{rtlpp}
M.~Akyash, K.~Azar, and H.~Kamali, ``{RTL++: Graph-Enhanced LLM for RTL Code
  Generation},'' in \emph{Proc.~ICLAD}, 2025.

\bibitem{mage}
Y.~Zhao, H.~Zhang, H.~Huang, Z.~Yu, and J.~Zhao, ``{MAGE: A Multi-Agent Engine
  for Automated RTL Code Generation},'' in \emph{Proc.~DAC}, 2025.

\bibitem{spec2rtl}
Z.~Yu, M.~Liu, M.~Zimmer, Y.~Celine, Y.~Liu, and H.~Ren, ``{Spec2RTL-Agent:
  Automated Hardware Code Generation from Complex Specifications Using LLM
  Agent Systems},'' in \emph{Proc.~ICLAD}, 2025.

\bibitem{symrtlo}
Y.~Wang, W.~Ye, P.~Guo, Y.~He, Z.~Wang, B.~Tian, S.~He, G.~Sun, Z.~Shen,
  S.~Chen \emph{et~al.}, ``{SymRTLO: Enhancing RTL Code Optimization with LLMs
  and Neuron-Inspired Symbolic Reasoning},'' \emph{Proc.~NIPS}, 2026.

\bibitem{rtlrewriter}
X.~Yao, Y.~Wang, X.~Li, Y.~Lian, R.~Chen, L.~Chen, M.~Yuan, H.~Xu, and B.~Yu,
  ``{RTLRewriter: Methodologies for Large Models Aided RTL Code
  Optimization},'' in \emph{Proc.~ICCAD}, 2024.

\bibitem{aspen}
N.~Zhang, C.~Deng, J.~M. Kuehn, C.-T. Ho, C.~Yu, Z.~Zhang, and H.~Ren,
  ``{ASPEN: LLM-Guided E-Graph Rewriting for RTL Datapath Optimization},'' in
  \emph{Proc.~MLCAD}, 2025.

\bibitem{learntodebug}
Y.~Bai and H.~Ren, ``{Learning to Debug: LLM-Organized Knowledge Trees for
  Solving RTL Assertion Failures},'' \emph{arXiv preprint}, 2025.

\bibitem{hlsdebugger}
J.~Wang, S.~Liu, Y.~Lu, and Z.~Xie, ``{HLSDebugger: Identification and
  Correction of Logic Bugs in HLS Code with LLM Solutions},'' in
  \emph{Proc.~ICCAD}, 2025.

\bibitem{assertllm}
Z.~Yan, W.~Fang, M.~Li, M.~Li, S.~Liu, Z.~Xie, and H.~Zhang, ``{Assertllm:
  Generating hardware verification assertions from design specifications via
  multi-llms},'' in \emph{Proc.~ASPDAC}, 2025.

\bibitem{mlir}
C.~Lattner, M.~Amini, U.~Bondhugula, A.~Cohen, A.~Davis, J.~Pienaar, R.~Riddle,
  T.~Shpeisman, N.~Vasilache, and O.~Zinenko, ``{MLIR: Scaling compiler
  infrastructure for domain specific computation},'' in \emph{Proc.~CGO}, 2021.

\bibitem{assassyn}
J.~Weng, B.~Han, D.~Gao, R.~Gao, W.~Zhang, A.~Zhong, C.~Xu, J.~Xin, Y.~Luo,
  L.~W. Wills \emph{et~al.}, ``{Assassyn: A Unified Abstraction for
  Architectural Simulation and Implementation},'' in \emph{Proc.~ISCA}, 2025.

\bibitem{cement}
Y.~Xiao, Z.~Luo, K.~Zhou, and Y.~Liang, ``{Cement: Streamlining FPGA Hardware
  Design with Cycle-Deterministic EHDL and Synthesis},'' in \emph{Proc.~FPGA},
  2024.

\bibitem{chisel}
J.~Bachrach, H.~Vo, B.~Richards, Y.~Lee, A.~Waterman, R.~Avi{\v{z}}ienis,
  J.~Wawrzynek, and K.~Asanovi{\'c}, ``{Chisel: constructing hardware in a
  scala embedded language},'' in \emph{Proc.~DAC}, 2012.

\bibitem{calyx}
R.~Nigam, S.~Thomas, Z.~Li, and A.~Sampson, ``{A Compiler Infrastructure For
  Accelerator Generators},'' in \emph{Proc.~ASPLOS}, 2021.

\bibitem{pipertl}
S.~Yin, F.~Liu, L.~Zou, R.~Fu, W.~Zhao, C.~Bai, T.-Y. Ho, Y.~Xie, and B.~Yu,
  ``{PipeRTL: Timing-Aware Pipeline Optimization at IR-Level for RTL
  Generation},'' \emph{arXiv preprint}, 2026.

\bibitem{combrewriter}
H.~Zheng, Z.~He, S.~Yin, Y.~Ma, and B.~Yu, ``{CombRewriter: Enabling
  Combinational Logic Simplification in MLIR-Based Hardware Compiler},'' in
  \emph{Proc.~ASPDAC}, 2026.

\bibitem{llhd}
F.~Schuiki, A.~Kurth, T.~Grosser, and L.~Benini, ``{LLHD: A Multi-level
  Intermediate Representation For Hardware Description Languages},'' in
  \emph{Proc.~PLDI}, 2020.

\bibitem{Khronos}
K.~Zhou, Y.~Liang, Y.~Lin, R.~Wang, and R.~Huang, ``{Khronos: Fusing Memory
  Access for Improved Hardware RTL Simulation},'' in \emph{Proc.~MICRO}, 2023.

\bibitem{llvm}
C.~Lattner and V.~Adve, ``{LLVM: A compilation framework for lifelong program
  analysis \& transformation},'' in \emph{Proc.~CGO}, 2004.

\bibitem{llm4ir}
H.~Jiang, J.~Zhu, Y.~Wan, B.~Fang, H.~Zhang, R.~Jin, and Q.~Guan, ``{Can Large
  Language Models Understand Intermediate Representations in Compilers?}'' in
  \emph{Proc.~ICML}, 2025.

\bibitem{appl}
H.~Dong, Q.~Su, Y.~Gao, Z.~Li, Y.~Ruan, G.~Pekhimenko, C.~J. Maddison, and
  X.~Si, ``{Appl: A prompt programming language for harmonious integration of
  programs and large language model prompts},'' in \emph{Proc.~ACL}, 2025.

\bibitem{sglang}
L.~Zheng, L.~Yin, Z.~Xie, C.~Sun, J.~Huang, C.~H. Yu, S.~Cao, C.~Kozyrakis,
  I.~Stoica, J.~E. Gonzalez \emph{et~al.}, ``{Sglang: Efficient execution of
  structured language model programs},'' \emph{Proc.~NIPS}, 2024.

\bibitem{vllm}
W.~Kwon, Z.~Li, S.~Zhuang, Y.~Sheng, L.~Zheng, C.~H. Yu, J.~Gonzalez, H.~Zhang,
  and I.~Stoica, ``{Efficient memory management for large language model
  serving with pagedattention},'' in \emph{Proc.~SOSP}, 2023.

\bibitem{llamaindex}
``{LlamaIndex},'' \url{https://github.com/jerryjliu/llama_index}, 2022.

\bibitem{codeeval}
M.~Chen, J.~Tworek, H.~Jun, Q.~Yuan, H.~P. D.~O. Pinto, J.~Kaplan, H.~Edwards,
  Y.~Burda, N.~Joseph, G.~Brockman \emph{et~al.}, ``{Evaluating large language
  models trained on code},'' \emph{arXiv preprint}, 2021.

\bibitem{iverilog}
S.~Williams, ``{The ICARUS Verilog Compilation System},''
  \url{https://steveicarus.github.io/iverilog/}, 2002.

\bibitem{yosys}
C.~Wolf, J.~Glaser, and J.~Kepler, ``{Yosys-a free verilog synthesis suite},''
  \url{https://yosyshq.net/yosys/}, 2013.

\bibitem{logicnn}
X.~Li, X.~Li, L.~Chen, X.~Zhang, M.~Yuan, and J.~Wang, ``{Logic synthesis with
  generative deep neural networks},'' \emph{arXiv preprint}, 2024.

\end{thebibliography}
}

\end{document}